\documentclass[10pt,envcountsame,envcountsect]{llncs}

\usepackage{latexsym,amsmath,amssymb}
\usepackage[dvips]{graphicx}
\usepackage[dvips]{color}

\newcommand{\rw}[2]{\ {\underset{#1}{\longrightarrow}}\,_{#2}\ }

\newcommand{\hra}[2]{\hookrightarrow_{#1}^{#2}}
\newcommand{\ra}[2]{\rightarrow_{#1}^{#2}}

\newcommand{\sra}[2]{\mbox{$\rightarrow\!${\tiny\raisebox{0.5ex}{$|$}}}_{#1}^{#2}}
\newcommand{\ehra}[2]{{\overset{#2}{\underset{#1}\hookrightarrow}}}
\newcommand{\era}[2]{{\overset{#2}{\underset{#1}\longrightarrow}}}

\newcommand{\esra}[2]{\overset{#2}
  {\underset{#1}{\mbox{$\longrightarrow\!${\tiny\raisebox{0.5ex}{$|$}}}}}}

\newcommand{\lra}{\longrightarrow} \newcommand{\Lra}{\Longrightarrow}
 
\newcommand{\ssq}{\mbox{\footnotesize$\square$}}

\newcommand{\ru}{R_+}
\newcommand{\rr}{R_=}
\newcommand{\rd}{R_-}

\newcommand{\nt}[1]{<\negthickspace#1\negthickspace>}

\newcommand{\piu}[1]{\pi_1(<\negthickspace#1\negthickspace>)}
\newcommand{\pid}[1]{\pi_2(<\negthickspace#1\negthickspace>)}

\author{Antoine Meyer \thanks{This work has been supported in part by
    the European IST-FET project ADVANCE (contract No.
    IST-1999-29082).}}
\institute{I{\sc risa}, campus de Beaulieu, Rennes\\
  L{\sc iafa}, Universit\'e de Paris 7\\
  \email{Antoine.Meyer@liafa.jussieu.fr}}
\title{On Term Rewriting Systems\\
  Having a Rational Derivation}

\begin{document}

\maketitle

\begin{abstract}
  Several types of term rewriting systems can be distinguished by
  the way their rules overlap. In particular, we define the classes of
  prefix, suffix, bottom-up and top-down systems, which generalize
  similar classes on words. Our aim is to study the derivation
  relation of such systems (i.e.  the reflexive and transitive closure
  of their rewriting relation) and, if possible, to provide a finite
  mechanism characterizing it. Using a notion of rational relations
  based on finite graph grammars, we show that the derivation of any
  bottom-up, top-down or suffix systems is rational, while it can be
  non recursive for prefix systems.
\end{abstract}

\section{Introduction}

Word rewriting systems are among the most general formalisms found in
computer science to model word transformations. They generalize
grammars, can represent the runs of finite automata, transducers,
pushdown automata or even Turing machines. They can thus be considered
as a unifying framework to compare all these heterogeneous formalisms.
For instance, \cite{ck02} proposes a homogeneous presentation of
several well-known families of infinite graphs, using an approach
based on word rewriting systems proposed in \cite{ck98}, which is to
consider the `Cayley graph' of a rewriting system. In another paper
\cite{ca00}, a classification of word rewriting systems according to
the way their rules overlap is established. It is proved that the
derivation relations of four classes of systems are rational, which
means that they can be generated by finite transducers. These systems
called left, right, prefix and suffix, were later used in \cite{ck02}.
Any other class is shown to contain at least one system whose
derivation is not rational.

The aim of this work is to extend these results from words to terms.
To summarize, we will be interested in term rewriting systems whose
derivation can be characterized by a finite mechanism. First of all,
we have to specify which definition of rationality for relations on
terms we intend to use, as several distinct notions already exist (see
\cite{ra91} for an overview).  Unfortunately, none of them is as
widely adopted as the standard one for words, as each relies on
different characteristics of the word case, and serves a different
purpose. In this paper, we will adopt a notion introduced in
\cite{ra97}, which makes use of hyperedge replacement graph grammars.
The reason for this choice is the close similarity between the way
these grammars work, and the asynchronous mechanism of a word
transducer. Then we extend the definitions of left, right, prefix and
suffix systems to terms, yielding what we will call bottom-up,
top-down, prefix and suffix systems, and investigate the rationality
of their derivation relations. We also mention recognizability
preservation properties for bottom-up, top-down and suffix systems.

Numerous works deal with term rewriting systems. Among the closest to
our approach, we can mention for instance \cite{gv98} and
\cite{tks00}, which specifically investigate the recognizability
preservation properties of term rewriting systems. Both papers study
classes of systems which properly include the class of top-down
systems, and prove that they preserve recognizability.  However, the
derivation relations of these systems are not rational (more
generally, no finite representation of these relations is given). On
the contrary, Dauchet and Tison extensively studied ground term
rewriting systems, i.e. systems whose rules do not contain variables
\cite{dt85}. In particular, they proved that these systems have a
decidable first order theory with reachability \cite{dt90} by
explicitely building their derivation relation. From another point of
view, \cite{lo02} and \cite{co02} investigated the geometric
properties of transition graphs of ground systems and compared this
family of graphs with respect to other well-known families. Note that
by definition, ground systems are a special kind of suffix systems.
Finally, we can mention the theme of symbolic model-checking, whose
main idea is to represent regular sets of configurations by finite
word automata and system transitions by rewrite rules or transducers
(see for example \cite{bjnt00}). This field is currently being
extended to systems with richer topologies, like trees
\cite{ajmo02,bt02}. A central problem relevant to this method is to
compute the set of configurations reachable in any number of steps
when starting from a \emph{regular} set of configurations (for
instance a recognizable term language).

This paper is organized as follows: after recalling a few basic
notions about trees, terms and recognizable languages, we present the
notion of rationality for term relations introduced in \cite{ra97}. In
Section \ref{sec:trs} we introduce term rewriting systems, and detail
the four subclasses we consider. The last two parts present our
results concerning the rationality of the derivations of top-down,
bottom-up and suffix systems, as well as remarks concerning their
preservation of recognizability.

\section{Terms and Trees}
\label{prelim}

Let $F = \bigcup_{n \geq 0} F_n$ be a finite ranked alphabet, each
$F_n$ being a set of \emph{function} symbols of arity $n$ (elements of
$F_0$ are \emph{constants}), and $X$ be a finite set of
\emph{variable} symbols.  The set $T(F,X)$ of finite first-order terms
on $F$ with variables in $X$ is the smallest set including $X$ and
satisfying $f \in F_n\ \land\ t_1,\ldots,t_n \in T(F,X)\ \Rightarrow\ 
f t_1 \ldots t_n \in T(F,X)$.  The set $T(F,X)^+$ of tuples of terms
will be called the set of \emph{term words}. A term word $t =
(t_1,\ldots ,t_n)$ is usually noted $t_1 \ldots t_n$, and $t(i)$ is
used to denote $t_i$. The dimension of $t$ is called its length and
noted $|t|$ (here $|t| = n$). Term words containing no variable are
called \emph{ground}. The set of ground terms is noted
$T(F,\emptyset)$ or simply $T(F)$. The set of variables actually
occurring in a term or term word $t$ is $\mathit{Var}(t)$, and $t$ is
said \emph{linear} if each of its variables occurs only once. If
moreover $t$ has $n$ variables, it is called a \emph{n-context}. The
variables of a $n$-context are conventionally noted $\ssq_1, \ldots
\ssq_n$. The set of $n$-contexts is denoted by $C_n(F)$, the set of
all contexts by $C(F)$. A common operation on terms is
\emph{substitution}. A substitution is fully defined by a mapping from
$X$ to $T(F,X)$, and extended to a morphism as follows: we note $t
\sigma$ the application of a substitution $\sigma$ to a term word $t$,
which is done by replacing every occurrence of each variable $x$
occurring in $t$ by the term $\sigma(x)$. The set of substitutions
over $F$ and $X$ is noted $S(F,X)$. For any term word $s = s_1 \ldots
s_n$ and when $t$ is a $n$-context, we use $t[s]$ as a shorthand
notation for the variable substitution $t \{\ssq_i \mapsto s_i\ |\ i
\in [1,n]\}$. All these notations are extended to sets of term words
in the usual way. A term, term word, context or substitution is said
to be \emph{proper} or \emph{non-trivial} if it contains at least one
symbol in $F$.

Let $\mathbb{N}$ be the set of strictly positive integers, we call
\emph{position} any word in the set $\mathbb{N}^*$. Every term $t$ in
$T(F,X)$ can be represented as a finite ordered tree whose nodes are
labeled by symbols in $F$ or variables in $X$, or equivalently as a
mapping from a prefix-closed set of \emph{positions}
$\mathit{Pos}(t)$, called the \emph{domain} of the term, to the set $F
\cup X$. Let $t = f(t_1,\ldots,t_n,\ldots)$ be a term represented by
an ordered tree, position $\varepsilon$ denotes the root of $t$, and
for $n \in \mathbb{N},\ p \in \mathbb{N}^*$, position $np$ denotes the
node at position $p$ in subtree $t_n$. Seeing terms as trees, term
words can be seen as \emph{ordered forests}. In the following, we will
use the prefix partial order on positions, noted $\geqslant$: let $p$
and $q$ be two positions, $p \geqslant q$ if there is some $q' \in
\mathbb{N}^*$ such that $p = qq'$. If furthermore $q' \neq
\varepsilon$, we write $p > q$. We denote by $\mathit{pos}(x,t)$ the
set of positions at which the variable $x \in X$ occurs in term $t \in
T(F,X)$.

The most common acceptors for languages of trees (and thus terms) are
finite tree automata. Among several variants, we will only consider
\emph{top-down} tree automata, defined by a finite set $Q$ of control
states and a finite set $R$ of transition rules of the form $qf
\rightarrow f q_1 \ldots q_n$ where $f \in F_n$ ($n$ can be $0$) and
$q, q_1, \ldots q_n \in Q$. A configuration is an embedding of control
states in the input tree, i.e. a tree from $T(F\,\cup\,Q,X)$, where
each $q \in Q$ is considered a unary function symbol. A rule $qf
\rightarrow f q_1 \ldots q_n$ can be applied in configuration $c_1$ to
reach configuration $c_2$ if $c_1 = t[q f t_1 \ldots t_n]$ and $c_2 =
t[f q_1 t_1 \ldots q_n t_n]$ for any context $t$ and term word $t_1
\ldots t_n$. A run of the automaton is a sequence of applications of
rules on a given input. A ground term $t \in T(F,X)$ is
\emph{accepted} or \emph{recognized} if there is a run from
configuration $q_0 t$ (where $q_0$ is an initial control state) to
configuration $t$. The set of terms accepted by a tree automaton $A$
is called the language of $A$ and noted $L(A)$. The languages accepted
by finite tree automata are called \emph{recognizable}.

\section{Rational Tree Relations}

Several authors have tried to define suitable notions of binary
relations over terms generalizing known families of relations over
words, like for instance the recognizable relations, or the more
general rational relations (i.e. relations recognized by finite
transducers). As of now, no extension to terms is really considered
canonical, as each family of relations has its own merits and
drawbacks. Several distinct families can be encountered: recognizable
relations as such, relations defined as rational languages over some
overlap coding of both projections of the relation, relations induced
by various types of tree transducers, or the more specific class of
ground tree transductions, to cite but a few (see \cite{ra91} for a
survey).

In \cite{ra97}, a notion of rationality for tuples of trees according
to the union, substitution and iterated substitution operations is
proposed. This notion can also be seen as a definition for binary
rational relations over tuples of trees, and thus as a special case,
binary relations over trees. Similarly to the word case, this class is
strictly more general than the class of recognizable relations. In
his paper, Raoult proves that the rational languages of tuples can be
generated using a special kind of hyperedge replacement grammars. This
definition is justified by its similarity to rational word relations
on several aspects: first, as it should be, it coincides with rational
word relations when restricted to trees of degree one. Second, it is
closed under projection on any number of components, union and
intersection. Finally, its mechanism is indeed quite close to the way
a transducer works. However, this generality has a cost, and this
class of relations is not closed under composition.

First, we need to define the product operation we shall use to define
rational sets, which is an extension of the usual substitution
operation. Let $t$ be a term word, $x$ a word of $n$ variables having
$k \geq 0$ instances $x_1 \ldots x_k$ in $t$ (i.e. a total of $n * k$
variables), and $M$ a set of $n$-tuples of terms. We define $t \cdot_x
M$ as the set of tuples of terms obtained by replacing each instance
of $x$ in $t$ with a (possibly different) element of $M$. Formally: $t
\cdot_x M := \{t \{ s_i(j) \mapsto x_i(j)\,|\,i \in [1,k],\,j \in
[1,n] \}\}$.  It is extended to sets in the usual way: $L \cdot_x M :=
\{t \cdot_x M\ |\ t \in L\}$. Furthermore, define $L^{n_x} := L
\cdot_x L^{{n-1}_x}$ and $L^{*_x} := \bigcup_{n \geq 0} L^{n_x}$. We
are now ready to define the notion of rationality associated to this
product:

\begin{definition}[\cite{ra97}]
  The set $\mathit{Rat}_n$ of rational languages of $n$-tuples of
  trees is the smallest set of languages containing the finite
  languages of tuples and closed under the following operations:
  \begin{enumerate}
  \item $L \in \mathit{Rat}_n\ \land\ M \in \mathit{Rat}_n\ 
    \Rightarrow\ L \cup M \in \mathit{Rat}_n$
  \item $L \in \mathit{Rat}_n\ \land\ x \in X^m\ \land\ M \in
    \mathit{Rat}_m\ \Rightarrow\ L \cdot_{x} M \in \mathit{Rat}_n$
  \item $L \in \mathit{Rat}_n\ \land\ x \in X^n\ \Rightarrow\ 
    L^{*_{x}} \in \mathit{Rat}_n$
  \end{enumerate}
  The family $\mathit{Rat}$ of rational languages over tuples of terms
  is the union of all $\mathit{Rat}_n$, for $n \geq 1$.
\end{definition}

One should note that this notion of rationality differs from the one
defined in \cite{lw01}, for example, as the concatenation (or
`series') product is not directly taken into account, and substitution
is done simultaneously on several variables. From this definition
arises a straightforward notion of \emph{rational expression}, which
extends the usual notion on words. It should be noted that
$\mathit{Rat}_1$ does not coincide with the set of recognizable term
languages. For example, on $\Sigma = \{f^{(2)}, g^{(1)}, {h}^{(1)},
a^{(0)}\}$, the language $fg^nag^na \in \mathit{Rat}_1$ is defined by
the rational expression $f \ssq_1 \ssq_2 [g \ssq_1 g \ssq_2]^* [a a]$,
but it is not a recognizable term language.

Let us now recall the hyperedge replacement grammars used in
\cite{ra97}, which generate the rational languages of tuples of terms.
In this paper, we will call \emph{grammar} a hyperedge replacement
grammar such that every production $(A,\alpha)$ has the following
properties:
\begin{itemize}
\item the terminal subgraph of $\alpha$, say $\alpha_t$, obtained by
  removing all non-terminal hyperedges from $\alpha$, is an ordered
  forest with $n$ connex components (a $n$-tuple of trees), where $n$
  is the arity of $A$,
\item the vertices of $\alpha$ belonging to a hyperedge are leaves of
  $\alpha_t$,
\item no vertex of $\alpha$ belongs to more than one hyperedge.
\end{itemize}
These properties allow us to refer to the right-hand sides of this
type of grammars as `leaf-linked forests'. The definition of grammar
derivation is the usual one for hyperedge replacement. It will be
useful to also recall the formal definition of a grammar from the
point of view of terms, as it is done in the original paper:

\begin{definition}[\cite{ra97}]
  Given a set $X$ of variables, a \emph{production} is a pair
  $(A,\alpha)$, where $A \in X^n$ ($A = A_1 \ldots A_n$ is called a
  \emph{non-terminal}), $\alpha \in T(F,X \times \mathbb{N})^n$ (here
  $X \times \mathbb{N}$ denotes the set of numbered instances of
  variables of $X$), and both $A$ and $\alpha$ are linear. A
  \emph{grammar} is a finite set of productions such that the
  variables occurring in the right-hand sides can be grouped to form
  instances of non-terminals. A step of derivation of a grammar is
  defined as $t \ra{G}{} t \{A_j^i \mapsto \alpha_j\,|\,j \in [1,n]\}$
  where $t$ is a term word, there is a production $(A,\alpha)$ in $G$
  and $A^i$ is an instance of $A$ in $t$. The language generated by a
  grammar $G$ from axiom $A$ is the set of tuples of ground trees
  $L(G,A) = \{w \in T(F)^{|A|}\ |\ A \ra{G}{*} w\}$.
\end{definition}

\begin{example}
  \label{ex:gr}
  Let $A = A_1A_2$ and $B = B_1B_2B_3$ be two non-terminals of
  respective arity 2 and 3. The grammar $G_1$ having rules
  \begin{align*}
    A\ \lra &\ a\ a\ |\ gA_1\ gA_2\ |\ fA_1^1A_1^2\ fA_2^1A_2^2\ |\ 
    fB_1B_2\ B_3
    \\
    B\ \lra &\ A_1^1\ A_1^2\ fA_2^1A_2^2\ |\ gB_1\ gB_2\ hB_3
  \end{align*}
  can be represented as a HR grammar in the following way:
  
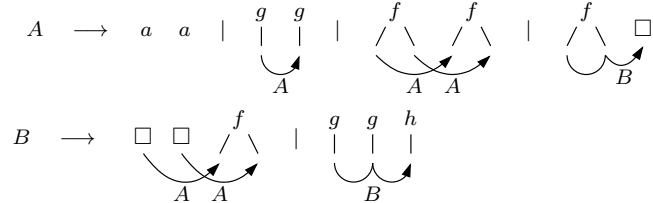
\begin{center}\begin{picture}(0,0)%
\includegraphics{graphics/ex_gr2.pstex}%
\end{picture}%
\setlength{\unitlength}{3947sp}%
\begingroup\makeatletter\ifx\SetFigFont\undefined%
\gdef\SetFigFont#1#2#3#4#5{%
  \reset@font\fontsize{#1}{#2pt}%
  \fontfamily{#3}\fontseries{#4}\fontshape{#5}%
  \selectfont}%
\fi\endgroup%
\begin{picture}(3453,1271)(1561,-2461)
\put(5002,-1394){\makebox(0,0)[b]{\smash{{\SetFigFont{8}{9.6}{\familydefault}{\mddefault}{\updefault}{\color[rgb]{0,0,0}$\ssq$}%
}}}}
\put(1642,-1394){\makebox(0,0)[rb]{\smash{{\SetFigFont{8}{9.6}{\familydefault}{\mddefault}{\updefault}{\color[rgb]{0,0,0}$A\quad\lra$}%
}}}}
\put(1882,-1394){\makebox(0,0)[b]{\smash{{\SetFigFont{8}{9.6}{\familydefault}{\mddefault}{\updefault}{\color[rgb]{0,0,0}$a$}%
}}}}
\put(2122,-1394){\makebox(0,0)[b]{\smash{{\SetFigFont{8}{9.6}{\familydefault}{\mddefault}{\updefault}{\color[rgb]{0,0,0}$a$}%
}}}}
\put(2362,-1394){\makebox(0,0)[b]{\smash{{\SetFigFont{8}{9.6}{\familydefault}{\mddefault}{\updefault}{\color[rgb]{0,0,0}$|$}%
}}}}
\put(3082,-1394){\makebox(0,0)[b]{\smash{{\SetFigFont{8}{9.6}{\familydefault}{\mddefault}{\updefault}{\color[rgb]{0,0,0}$|$}%
}}}}
\put(4282,-1394){\makebox(0,0)[b]{\smash{{\SetFigFont{8}{9.6}{\familydefault}{\mddefault}{\updefault}{\color[rgb]{0,0,0}$|$}%
}}}}
\put(2602,-1274){\makebox(0,0)[b]{\smash{{\SetFigFont{8}{9.6}{\familydefault}{\mddefault}{\updefault}{\color[rgb]{0,0,0}$g$}%
}}}}
\put(2842,-1274){\makebox(0,0)[b]{\smash{{\SetFigFont{8}{9.6}{\familydefault}{\mddefault}{\updefault}{\color[rgb]{0,0,0}$g$}%
}}}}
\put(3442,-1274){\makebox(0,0)[b]{\smash{{\SetFigFont{8}{9.6}{\familydefault}{\mddefault}{\updefault}{\color[rgb]{0,0,0}$f$}%
}}}}
\put(3922,-1274){\makebox(0,0)[b]{\smash{{\SetFigFont{8}{9.6}{\familydefault}{\mddefault}{\updefault}{\color[rgb]{0,0,0}$f$}%
}}}}
\put(4642,-1274){\makebox(0,0)[b]{\smash{{\SetFigFont{8}{9.6}{\familydefault}{\mddefault}{\updefault}{\color[rgb]{0,0,0}$f$}%
}}}}
\put(2722,-1754){\makebox(0,0)[b]{\smash{{\SetFigFont{8}{9.6}{\familydefault}{\mddefault}{\updefault}{\color[rgb]{0,0,0}$A$}%
}}}}
\put(3562,-1754){\makebox(0,0)[b]{\smash{{\SetFigFont{8}{9.6}{\familydefault}{\mddefault}{\updefault}{\color[rgb]{0,0,0}$A$}%
}}}}
\put(3802,-1754){\makebox(0,0)[b]{\smash{{\SetFigFont{8}{9.6}{\familydefault}{\mddefault}{\updefault}{\color[rgb]{0,0,0}$A$}%
}}}}
\put(4882,-1694){\makebox(0,0)[b]{\smash{{\SetFigFont{8}{9.6}{\familydefault}{\mddefault}{\updefault}{\color[rgb]{0,0,0}$B$}%
}}}}
\put(1561,-2070){\makebox(0,0)[rb]{\smash{{\SetFigFont{8}{9.6}{\familydefault}{\mddefault}{\updefault}{\color[rgb]{0,0,0}$B\quad\lra$}%
}}}}
\put(1861,-2070){\makebox(0,0)[b]{\smash{{\SetFigFont{8}{9.6}{\familydefault}{\mddefault}{\updefault}{\color[rgb]{0,0,0}$\ssq$}%
}}}}
\put(2101,-2070){\makebox(0,0)[b]{\smash{{\SetFigFont{8}{9.6}{\familydefault}{\mddefault}{\updefault}{\color[rgb]{0,0,0}$\ssq$}%
}}}}
\put(2821,-2070){\makebox(0,0)[b]{\smash{{\SetFigFont{8}{9.6}{\familydefault}{\mddefault}{\updefault}{\color[rgb]{0,0,0}$|$}%
}}}}
\put(2461,-1950){\makebox(0,0)[b]{\smash{{\SetFigFont{8}{9.6}{\familydefault}{\mddefault}{\updefault}{\color[rgb]{0,0,0}$f$}%
}}}}
\put(3061,-1950){\makebox(0,0)[b]{\smash{{\SetFigFont{8}{9.6}{\familydefault}{\mddefault}{\updefault}{\color[rgb]{0,0,0}$g$}%
}}}}
\put(3301,-1950){\makebox(0,0)[b]{\smash{{\SetFigFont{8}{9.6}{\familydefault}{\mddefault}{\updefault}{\color[rgb]{0,0,0}$g$}%
}}}}
\put(3301,-2430){\makebox(0,0)[b]{\smash{{\SetFigFont{8}{9.6}{\familydefault}{\mddefault}{\updefault}{\color[rgb]{0,0,0}$B$}%
}}}}
\put(2101,-2430){\makebox(0,0)[b]{\smash{{\SetFigFont{8}{9.6}{\familydefault}{\mddefault}{\updefault}{\color[rgb]{0,0,0}$A$}%
}}}}
\put(2341,-2430){\makebox(0,0)[b]{\smash{{\SetFigFont{8}{9.6}{\familydefault}{\mddefault}{\updefault}{\color[rgb]{0,0,0}$A$}%
}}}}
\put(3541,-1950){\makebox(0,0)[b]{\smash{{\SetFigFont{8}{9.6}{\familydefault}{\mddefault}{\updefault}{\color[rgb]{0,0,0}$h$}%
}}}}
\end{picture}%
\end{center}
  Then a possible production sequence of $G_1$ would be:
  
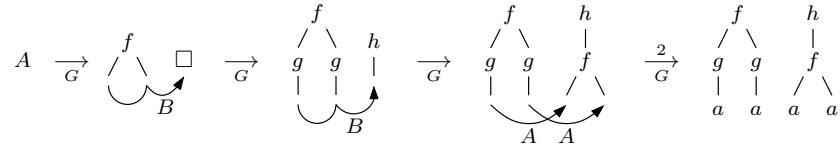
\begin{center}\begin{picture}(0,0)%
\includegraphics{graphics/ex_grsequ2.pstex}%
\end{picture}%
\setlength{\unitlength}{3947sp}%
\begingroup\makeatletter\ifx\SetFigFont\undefined%
\gdef\SetFigFont#1#2#3#4#5{%
  \reset@font\fontsize{#1}{#2pt}%
  \fontfamily{#3}\fontseries{#4}\fontshape{#5}%
  \selectfont}%
\fi\endgroup%
\begin{picture}(4747,891)(1703,-1962)
\put(1941,-1334){\makebox(0,0)[b]{\smash{{\SetFigFont{8}{9.6}{\familydefault}{\mddefault}{\updefault}{\color[rgb]{0,0,0}$f$}%
}}}}
\put(5771,-1158){\makebox(0,0)[b]{\smash{{\SetFigFont{8}{9.6}{\familydefault}{\mddefault}{\updefault}{\color[rgb]{0,0,0}$f$}%
}}}}
\put(4342,-1158){\makebox(0,0)[b]{\smash{{\SetFigFont{8}{9.6}{\familydefault}{\mddefault}{\updefault}{\color[rgb]{0,0,0}$f$}%
}}}}
\put(3132,-1155){\makebox(0,0)[b]{\smash{{\SetFigFont{8}{9.6}{\familydefault}{\mddefault}{\updefault}{\color[rgb]{0,0,0}$f$}%
}}}}
\put(2299,-1453){\makebox(0,0)[b]{\smash{{\SetFigFont{8}{9.6}{\familydefault}{\mddefault}{\updefault}{\color[rgb]{0,0,0}$\ssq$}%
}}}}
\put(2180,-1750){\makebox(0,0)[b]{\smash{{\SetFigFont{8}{9.6}{\familydefault}{\mddefault}{\updefault}{\color[rgb]{0,0,0}$B$}%
}}}}
\put(1703,-1453){\makebox(0,0)[rb]{\smash{{\SetFigFont{8}{9.6}{\familydefault}{\mddefault}{\updefault}{\color[rgb]{0,0,0}$A\quad\era{G}{}$}%
}}}}
\put(2656,-1453){\makebox(0,0)[b]{\smash{{\SetFigFont{8}{9.6}{\familydefault}{\mddefault}{\updefault}{\color[rgb]{0,0,0}$\era{G}{}$}%
}}}}
\put(3866,-1455){\makebox(0,0)[b]{\smash{{\SetFigFont{8}{9.6}{\familydefault}{\mddefault}{\updefault}{\color[rgb]{0,0,0}$\era{G}{}$}%
}}}}
\put(5295,-1455){\makebox(0,0)[b]{\smash{{\SetFigFont{8}{9.6}{\familydefault}{\mddefault}{\updefault}{\color[rgb]{0,0,0}$\era{G}{2}$}%
}}}}
\put(5652,-1455){\makebox(0,0)[b]{\smash{{\SetFigFont{8}{9.6}{\familydefault}{\mddefault}{\updefault}{\color[rgb]{0,0,0}$g$}%
}}}}
\put(5890,-1455){\makebox(0,0)[b]{\smash{{\SetFigFont{8}{9.6}{\familydefault}{\mddefault}{\updefault}{\color[rgb]{0,0,0}$g$}%
}}}}
\put(6247,-1455){\makebox(0,0)[b]{\smash{{\SetFigFont{8}{9.6}{\familydefault}{\mddefault}{\updefault}{\color[rgb]{0,0,0}$f$}%
}}}}
\put(5652,-1753){\makebox(0,0)[b]{\smash{{\SetFigFont{8}{9.6}{\familydefault}{\mddefault}{\updefault}{\color[rgb]{0,0,0}$a$}%
}}}}
\put(5890,-1753){\makebox(0,0)[b]{\smash{{\SetFigFont{8}{9.6}{\familydefault}{\mddefault}{\updefault}{\color[rgb]{0,0,0}$a$}%
}}}}
\put(6128,-1753){\makebox(0,0)[b]{\smash{{\SetFigFont{8}{9.6}{\familydefault}{\mddefault}{\updefault}{\color[rgb]{0,0,0}$a$}%
}}}}
\put(6366,-1753){\makebox(0,0)[b]{\smash{{\SetFigFont{8}{9.6}{\familydefault}{\mddefault}{\updefault}{\color[rgb]{0,0,0}$a$}%
}}}}
\put(4223,-1455){\makebox(0,0)[b]{\smash{{\SetFigFont{8}{9.6}{\familydefault}{\mddefault}{\updefault}{\color[rgb]{0,0,0}$g$}%
}}}}
\put(4461,-1455){\makebox(0,0)[b]{\smash{{\SetFigFont{8}{9.6}{\familydefault}{\mddefault}{\updefault}{\color[rgb]{0,0,0}$g$}%
}}}}
\put(4818,-1455){\makebox(0,0)[b]{\smash{{\SetFigFont{8}{9.6}{\familydefault}{\mddefault}{\updefault}{\color[rgb]{0,0,0}$f$}%
}}}}
\put(4461,-1931){\makebox(0,0)[b]{\smash{{\SetFigFont{8}{9.6}{\familydefault}{\mddefault}{\updefault}{\color[rgb]{0,0,0}$A$}%
}}}}
\put(4699,-1931){\makebox(0,0)[b]{\smash{{\SetFigFont{8}{9.6}{\familydefault}{\mddefault}{\updefault}{\color[rgb]{0,0,0}$A$}%
}}}}
\put(3013,-1453){\makebox(0,0)[b]{\smash{{\SetFigFont{8}{9.6}{\familydefault}{\mddefault}{\updefault}{\color[rgb]{0,0,0}$g$}%
}}}}
\put(3251,-1453){\makebox(0,0)[b]{\smash{{\SetFigFont{8}{9.6}{\familydefault}{\mddefault}{\updefault}{\color[rgb]{0,0,0}$g$}%
}}}}
\put(3370,-1869){\makebox(0,0)[b]{\smash{{\SetFigFont{8}{9.6}{\familydefault}{\mddefault}{\updefault}{\color[rgb]{0,0,0}$B$}%
}}}}
\put(3489,-1334){\makebox(0,0)[b]{\smash{{\SetFigFont{8}{9.6}{\familydefault}{\mddefault}{\updefault}{\color[rgb]{0,0,0}$h$}%
}}}}
\put(4818,-1158){\makebox(0,0)[b]{\smash{{\SetFigFont{8}{9.6}{\familydefault}{\mddefault}{\updefault}{\color[rgb]{0,0,0}$h$}%
}}}}
\put(6247,-1158){\makebox(0,0)[b]{\smash{{\SetFigFont{8}{9.6}{\familydefault}{\mddefault}{\updefault}{\color[rgb]{0,0,0}$h$}%
}}}}
\end{picture}%
\end{center}
\end{example}

\noindent
As expected, these grammars generate the rational languages:

\begin{theorem}[\cite{ra97}]
  A language of tuples of terms is rational if and only if it is the
  language generated by a grammar.
\end{theorem}

A rational language of $n$-tuples of terms can also be seen as a
binary relation in $T(F)^p \times T(F)^q$, where $p+q = n$. In this
case, given a non-terminal $A$, we define the first and second
projections $\pi_1(A)$ and $\pi_2(A)$ by the set of variables of $A$
referring to the first (resp. second) projection of the relation. A
similar notation is used for right-hand sides of grammar productions
as well.  For clarity, we write a production $(A,\alpha)$ as
$(A,\pi_1(\alpha) \times \pi_2(\alpha))$. Without loss of generality,
we always consider that $A = \pi_1(A)\pi_2(A)$ and $\alpha =
\pi_1(\alpha)\pi_2(\alpha)$. For example, if the axiom of a $T(F)^p
\times T(F)^q$ relation is $A = A_1 \ldots A_n$, we can have $\pi_1(A)
= A_1 \ldots A_p$ and $\pi_2(A) = A_{p+1} \ldots A_{p+q = n}$.

\begin{example}
  \label{ex:gr_rel}
  Grammar $G_1$ from Ex. \ref{ex:gr} generates, from non-terminal $A$,
  a language $L(G_1,A)\ \in\ \mathit{Rat}_2$, which can be seen as a
  $T(F) \times T(F)$ relation. In this case, its rules can be written
  \begin{align*}
    A\ \lra &\ a\ \times\ a\ |\ gA_1\ \times\ gA_1\ |\ fA_1^1A_1^2\ 
    \times\ fA_2^1A_2^2\ |\ fB_1B_2\ \times\ B_3
    \\
    B\ \lra &\ A_1^1\ A_1^2\ \times\ fA_2^1A_2^2\ |\ gB_1\ gB_2\ 
    \times\ hB_3
  \end{align*}
\end{example}

\section{Term Rewriting Systems}
\label{sec:trs}

A (term) \emph{rewrite rule} is a pair $(l,r) \in T(F,X)^2$ such that
$\mathit{Var}(r)\ \subseteq \ \mathit{Var}(l)$. A rewrite rule $(l,r)$
is said to be \emph{linear} if both $l$ and $r$ are.  A \emph{rewrite
  system}, or more specifically \emph{term rewriting system} is a set
of rewrite rules $R$. A system $R$ is \emph{finite} when $|R|$ is
finite, and \emph{recognizable} when the potentially infinite number
of rules is given as a finite union of pairs $U \rightarrow V$, where
$U$ and $V$ are recognizable term languages. Note that we only
consider systems where the total number of distinct variables is
finite. A system is \emph{linear} when all its rules are linear.  We
denote by $\mathit{Dom}(R)$ (resp. $\mathit{Ran}(R)$) the set of
left-hand sides (resp. right-hand sides) of $R$, up to a renaming of
the variables. The \emph{rewriting} according to a system $R$ is the
relation
$$
\era{R}{}\ :=\ \{ (c[l\sigma],c[r\sigma]) \in T(F) \times T(F)\ |\ 
(l,r) \in R\ \land\ c \in C_1(F)\ \land\ \sigma \in S(F,X)\}.
$$
In case we want to specify that a rule $(l,r)$ is used at some
position $p$ (resp. set of positions $P$), we use the notation
$\rw{l,r}{p}$ (resp.  $\rw{l,r}{P}$). The reflexive and transitive
closure of $\ra{R}{}$ by composition is called the \emph{derivation}
of $R$ and written $\ra{R}{*}$.

\paragraph{Classification of Rewriting Systems.}
In the case of words, several natural classes of rewriting systems can
be distinguished by the way their rules are allowed to overlap. In
\cite{ca00}, the composition $\ra{R}{}\ \circ\ \ra{R}{}$ of two
rewritings is considered, and all the different possibilities of
overlapping between the right-hand side of the first rewrite rule, and
the left-hand side of the second one are examined. By discarding
systems where unwanted overlappings occur, one obtains four general
families of systems whose derivation is proven rational, the families
of \emph{left}, \emph{right}, \emph{prefix} and \emph{suffix} word
rewriting systems.  Moreover, any system which does not belong to one
of these families may have a non-rational derivation. As a
consequence, as terms generalize words, we only need to study the
extension of these four families of systems to terms: the classes of
\emph{bottom-up} and \emph{top-down} systems, which respectively
correspond to left and right systems, and the families of
\emph{prefix} and \emph{suffix} systems. A term rewriting system $R$
(resp.  its inverse $R^{-1}$) is said:
\begin{itemize}
\item \emph{top-down} (resp.  \emph{bottom-up}) if any overlapping
  between a right-hand side $r$ and a left-hand side $l$ of $R$ (resp.
  $R^{-1}$) is such that $r = \bar{r}[o]$ and $l = o\lambda$ for some
  (possibly trivial) $1$-context $\bar{r}$ and substitution $\lambda$,
\item \emph{prefix} if any overlapping between a right-hand side $r$
  and a left-hand side $l$ of $R$ is such that $l = r \lambda$ or $r =
  l \rho$ for some possibly trivial substitutions $\lambda$ and
  $\rho$,
\item \emph{suffix} if any overlapping between a right-hand side $r$
  and a left-hand side $l$ of $R$ is such that $l = \bar{l}[r]$ or $r
  = \bar{r}[l]$ for some possibly trivial $1$-contexts $\bar{r}$ and
  $\bar{l}$.
\end{itemize}
The following picture illustrates these four kinds of overlappings:
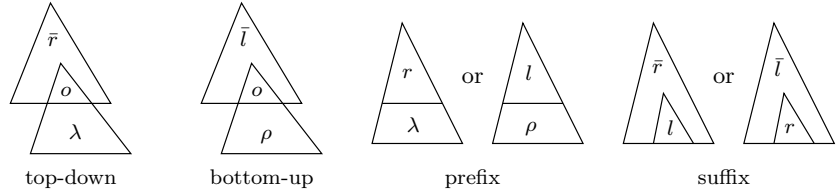
\begin{center}
  \begin{picture}(0,0)%
\includegraphics{graphics/trs2.pstex}%
\end{picture}%
\setlength{\unitlength}{3947sp}%
\begingroup\makeatletter\ifx\SetFigFont\undefined%
\gdef\SetFigFont#1#2#3#4#5{%
  \reset@font\fontsize{#1}{#2pt}%
  \fontfamily{#3}\fontseries{#4}\fontshape{#5}%
  \selectfont}%
\fi\endgroup%
\begin{picture}(5199,1184)(214,-483)
\put(3129,-447){\makebox(0,0)[b]{\smash{{\SetFigFont{8}{9.6}{\familydefault}{\mddefault}{\updefault}{\color[rgb]{0,0,0}prefix}%
}}}}
\put(4707,-447){\makebox(0,0)[b]{\smash{{\SetFigFont{8}{9.6}{\familydefault}{\mddefault}{\updefault}{\color[rgb]{0,0,0}suffix}%
}}}}
\put(4707,184){\makebox(0,0)[b]{\smash{{\SetFigFont{8}{9.6}{\familydefault}{\mddefault}{\updefault}{\color[rgb]{0,0,0}{\small or}}%
}}}}
\put(3129,184){\makebox(0,0)[b]{\smash{{\SetFigFont{8}{9.6}{\familydefault}{\mddefault}{\updefault}{\color[rgb]{0,0,0}{\small or}}%
}}}}
\put(2719,217){\makebox(0,0)[b]{\smash{{\SetFigFont{8}{9.6}{\familydefault}{\mddefault}{\updefault}{\color[rgb]{0,0,0}$r$}%
}}}}
\put(2757,-125){\makebox(0,0)[b]{\smash{{\SetFigFont{8}{9.6}{\familydefault}{\mddefault}{\updefault}{\color[rgb]{0,0,0}$\lambda$}%
}}}}
\put(3495,-114){\makebox(0,0)[b]{\smash{{\SetFigFont{8}{9.6}{\familydefault}{\mddefault}{\updefault}{\color[rgb]{0,0,0}$\rho$}%
}}}}
\put(3483,204){\makebox(0,0)[b]{\smash{{\SetFigFont{8}{9.6}{\familydefault}{\mddefault}{\updefault}{\color[rgb]{0,0,0}$l$}%
}}}}
\put(4290,246){\makebox(0,0)[b]{\smash{{\SetFigFont{8}{9.6}{\familydefault}{\mddefault}{\updefault}{\color[rgb]{0,0,0}$\bar{r}$}%
}}}}
\put(5048,228){\makebox(0,0)[b]{\smash{{\SetFigFont{8}{9.6}{\familydefault}{\mddefault}{\updefault}{\color[rgb]{0,0,0}$\bar{l}$}%
}}}}
\put(5122,-131){\makebox(0,0)[b]{\smash{{\SetFigFont{8}{9.6}{\familydefault}{\mddefault}{\updefault}{\color[rgb]{0,0,0}$r$}%
}}}}
\put(4372,-144){\makebox(0,0)[b]{\smash{{\SetFigFont{8}{9.6}{\familydefault}{\mddefault}{\updefault}{\color[rgb]{0,0,0}$l$}%
}}}}
\put(1773,103){\makebox(0,0)[b]{\smash{{\SetFigFont{8}{9.6}{\familydefault}{\mddefault}{\updefault}{\color[rgb]{0,0,0}$o$}%
}}}}
\put(574,103){\makebox(0,0)[b]{\smash{{\SetFigFont{8}{9.6}{\familydefault}{\mddefault}{\updefault}{\color[rgb]{0,0,0}$o$}%
}}}}
\put(492,418){\makebox(0,0)[b]{\smash{{\SetFigFont{8}{9.6}{\familydefault}{\mddefault}{\updefault}{\color[rgb]{0,0,0}$\bar{r}$}%
}}}}
\put(632,-168){\makebox(0,0)[b]{\smash{{\SetFigFont{8}{9.6}{\familydefault}{\mddefault}{\updefault}{\color[rgb]{0,0,0}$\lambda$}%
}}}}
\put(1691,418){\makebox(0,0)[b]{\smash{{\SetFigFont{8}{9.6}{\familydefault}{\mddefault}{\updefault}{\color[rgb]{0,0,0}$\bar{l}$}%
}}}}
\put(1831,-168){\makebox(0,0)[b]{\smash{{\SetFigFont{8}{9.6}{\familydefault}{\mddefault}{\updefault}{\color[rgb]{0,0,0}$\rho$}%
}}}}
\put(605,-447){\makebox(0,0)[b]{\smash{{\SetFigFont{8}{9.6}{\familydefault}{\mddefault}{\updefault}{\color[rgb]{0,0,0}top-down}%
}}}}
\put(1804,-447){\makebox(0,0)[b]{\smash{{\SetFigFont{8}{9.6}{\familydefault}{\mddefault}{\updefault}{\color[rgb]{0,0,0}bottom-up}%
}}}}
\end{picture}%

\end{center}
Prefix and suffix systems respectively generalize root and ground
rewriting systems. Root rewriting systems are already known to be very
powerful: indeed, they can simulate the execution steps of Turing
machines. This implies a direct negative result concerning prefix
systems.

\begin{proposition}
  Some linear prefix tree rewriting systems have a non rational
  derivation.
\end{proposition}

\begin{proof}
  Let $M$ be a Turing machine with a set of states $Q$, a tape
  alphabet $P$ and a set of transition rules $T \subseteq (Q \times P
  \cup \{\#\} \rightarrow Q \times P \times \{+,-\})$ ($\#$ denotes
  the `blank' character). Let us build a prefix system $R_M$ on the
  alphabet $Q \cup P \cup \{\#\}$, with variables in $\{x,y\}$, where
  $Q$ is considered binary, $P$ unary and $\#$ an overloaded symbol of
  arity either $0$ or $1$.\\
  For all $pA \rightarrow qB+\ \in\ T$, $R_M$ has a rule $pxAy
  \rightarrow qBxy$, plus a rule $px\# \rightarrow qBx\#$ if $A = \#$.
  \\
  For all $pA \rightarrow qB-\ \in\ T$ and $C\ \in\ P$, $R_M$ has
  rules $pCxAy \rightarrow qxCBy$ and $p\#Ay \rightarrow p\#\#By$,
  plus rules $p\#\# \rightarrow q\#\#B\#$ and $pCx\# \rightarrow
  qxCB\#$ if $A = \#$.
  \\
  This system has both overlappings of the kind $l = r\sigma$ and of
  the kind $r = l\sigma$, for some left and right-hand sides $l$ and
  $r$ and substitution $\sigma$. It is thus prefix, and neither
  top-down, bottom-up or suffix in general. It is quite clear that
  computing the derivation of $R_M$ is equivalent to computing the
  reachability relation of $M$, thus is undecidable. Hence $\ra{R}{*}$
  is non-recursive and can obviously not be rational. \qed
\end{proof}

However, contrary to the case of words, where prefix and suffix
systems are dual and share the same properties, the situation is
different in the case of terms. The family of ground rewriting
systems, which is a sub-family of suffix systems, has already been
studied by several authors. In particular, Dauchet and Tison
\cite{dt85} showed that the derivations of ground systems can be
recognized by a certain type of composite automata called \emph{ground
  tree transducers} (GTT). Section \ref{sec:suffix} will use similar
arguments in order to prove that, more generally, any suffix system
has a rational derivation. The two remaining families of term
rewriting systems we consider, namely top-down and bottom-up systems,
are dual. The next section puts focus on top-down systems, but all the
results extend to the bottom-up case (see Corollary \ref{cor:asc}).

\section{Derivation of Bottom-Up and Top-Down Systems}
\label{sec:desc}

This section focuses on the study of top-down term rewriting systems
and their derivations. For any finite linear top-down system, a
grammar of tuples of terms generating its derivation relation can be
built, which implies that this relation is rational. Furthermore, from
the shape of the grammar, we observe that the derivation of such a
system preserves the recognizability of term languages. Dual results
can be obtained for bottom-up systems: the derivation of a linear
bottom-up system is rational, and the inverse image of a recognizable
term language is still recognizable.

Let us first observe that top-down systems enjoy a kind of
\emph{monotonicity} feature. Any rewriting sequence of such systems is
equivalent to a sequence where the successive rewriting steps occur at
non-decreasing positions in the input term. We call this
\emph{top-down rewriting}. Let $R$ be a term rewriting system, we
define its top-down rewriting $\hra{R}{*}$ by:
$$
\ehra{R}{*}\ =\ \bigcup_{n \geq 0} \ehra{R}{n} \mbox{ with }
\begin{cases}
  \ehra{R}{0}\ =\ \mathit{Id}_{T(F)} \\[2ex]
  \ehra{R}{n}\ =\ \bigcup_{p_1,\ldots,p_n} \rw{u_1,v_1}{p_1} \circ
  \ldots \circ \rw{u_n,v_n}{p_n}
\end{cases}
$$
such that the rewriting positions do not decrease along indexes
($\forall i,j,\ i<j\ \Rightarrow\ \lnot (p_j < p_i)$), and if two
successive positions are equal then the second rewriting should not
have a trivial left-hand side ($(p_i = p_{i-1})\ \Rightarrow\ (u_i
\notin X)$). This last condition means that, for instance, the
sequence
$$
c[l\sigma] \rw{l,r}{} c[r\sigma] \rw{x,r'}{} c[r'\{x \mapsto
r\sigma\}]
$$
is not top-down, because the second rule produces its right-hand
side `higher' than the first one. The rewriting steps should be
swapped to obtain the top-down sequence
$$
c[l\sigma] \rw{x,r'}{} c[r'\{x \mapsto l\sigma\}]
\rw{l,r}{\mathit{pos}(x,c[r'])} c[r'\{x \mapsto r\sigma\}].
$$
The next lemma expresses the fact that, given any rewriting
sequence of a top-down system, rewriting steps can always be ordered
into an equivalent top-down sequence.

\begin{lemma}
  \label{lem:top-down}
  The relations of derivation and top-down derivation of any
  top-down term rewriting system $R$ coincide: $ \ra{R}{*}\ =\ 
  \hra{R}{*}$.
\end{lemma}

We are now ready to prove the rationality of the derivation of any
top-down rewriting system. Using this property of top-down
systems, it is possible to build a grammar which directly generates
the derivation of any such system. This grammar mimics the way a
rational word transducer works, using its control state to keep in
memory a finite subterm already read or yet to produce.

\begin{theorem}
  \label{th:top-down}
  Every finite linear top-down term rewriting system $R$ has a
  rational derivation.
\end{theorem}

\begin{proof}
  Let $R$ be a finite linear top-down system. We denote by $O$ the
  set of all overlappings between left and right parts of rules of
  $R$:
  \begin{multline}
  O = \{\ t \in C_n(F,X)\ |\ \exists\ s \in C_1(F,X), \\ u \in
  T(F,X)^n, s[t] \in \mathit{Ran}(R)\ \land\ t[u] \in \mathit{Dom}(R)\ 
  \}.
  \end{multline}
  Remark that $\ssq$ belongs to $O$. We will now build a grammar $G$
  whose language is exactly the derivation of $R$.  Its finite set of
  non-terminals is $\{\nt{*}\} \cup Q$, where $Q = \{\ \nt{t}\ =\ 
  \nt{t}_1 \ldots \nt{t}_{n+1}\ |\ t \in O\ \cap\ C_n(F)\} $ and, for
  all $\nt{t} \in Q$, $\pid{t}$ is a single variable. The production
  rules of $G$ are of four types.
  \\
  Type (1): $\forall\ f \in F_n,$
  \begin{align*}
    \nt{\ssq} &\ \rightarrow\ f\nt{\ssq}^1_1 \ldots \nt{\ssq}^n_1\ 
    \times\ f\nt{\ssq}^1_2 \ldots \nt{\ssq}^n_2
    \\
    \nt{*} &\ \rightarrow\ f\nt{*}^1 \ldots \nt{*}^n\\
    \intertext{Type (2): $\forall\ t \in O \cap C_n(F),\ t[u] \in O
      \cap C_m(F),$} \nt{t} &\ \rightarrow\ u[\piu{t[u]}]\ \times\ 
    \pid{t[u]}
    \\
    \intertext{Type (3): $\forall\ t[u] \in O,\ t \in O \cap
      C_\ell(F)$ (necessarily $\{u_1, \ldots u_\ell\} \subseteq O$),}
    \nt{t[u]} &\ \rightarrow\ \piu{u_1} \ldots \piu{u_\ell} \times\ 
    t[\pid{u_1} \ldots \pid{u_\ell}]
    \\
    \intertext{Type (4): $\forall\ (t[u],s[v]) \in R,\ v = v_1 \ldots
      v_\ell$ (necessarily $\{t,v_1, \ldots v_n\} \subseteq O$),}
    \nt{t} &\ \rightarrow\ u \sigma\ \times\ s[\pid{v_1} \ldots
    \pid{v_\ell}]
  \end{align*}
  where $\sigma$ is a variable renaming such that for any variable $x$
  of $u$, $\sigma(x) =\ \nt{v_i}_j$ if $x$ is the $j$-th variable to
  appear in $v_i$ (from left to right), and $\sigma(x) =\ \nt{*}$ if
  $x$ does not appear in any of the $v_i$. Figure \ref{fig:grdesc}
  illustrates the four types of rules.
  
  Intuitively, the role of this substitution is to gather into the
  same non-terminal or hyperedge all the variables of $u$ belonging to
  the same $v_i$, while respecting the order in which these variables
  appear in $v_i$.  This way, a correct instantiation of non-terminals
  of $G$ is ensured. If a variable of $u$ does not appear at all in
  $v$, then it means that a whole input subtree is `discarded' by the
  rewriting rule being applied. Thus the grammar should accept any
  subtree to be generated at this position, which is the role of the
  unary non-terminal $\nt{*}$.
  
  For simplicity, we will only consider type (4) rules in which
  $t,v_1,\ldots v_n$ are \emph{maximal}. The other cases can be
  simulated by suitable finite compositions of rules of types (2), (3)
  and (4). \qed
\end{proof}

\begin{figure}[btp]
  \begin{center}
    \input{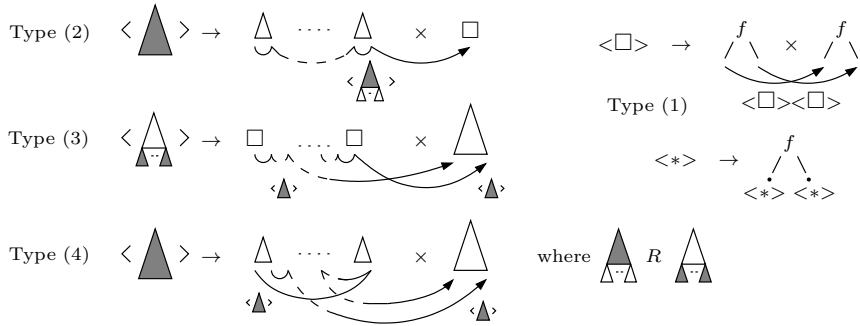}
  \end{center}
  {\caption{grammar associated to a top-down system.}
    \label{fig:grdesc}}
\end{figure}

\begin{example}
  \label{ex:desc}
  Consider the linear top-down system $R$ over the alphabet $F =
  \{f^{(2)},$ $g^{(1)}, {h}^{(1)}, a^{(0)}\}$ with a unique rule
  $fgxgy \rightarrow hfxy$. The corresponding grammar is the grammar
  of Ex. \ref{ex:gr_rel} where each non-terminal stands for one of the
  possible overlappings of rules of $R$: $A$ stands for $\ssq$ and $B$
  for $f\ssq_1\ssq_2$. Note that type (4) rules with non-maximal
  overlappings have been discarded. This example also illustrates the
  fact that the inverse image of a recognizable term language by the
  derivation of a linear top-down system is not recognizable in
  general: for instance, the image by $R_G^{-1}$ of ${h}^* f a a$ is
  $\{{h}^* f g^n a g^n a\ |\ n \geq 0\}$, which is not recognizable.
\end{example}

We will now mention a property of top-down systems, which has been
known for the past few years for larger classes of systems.

\begin{proposition}
  \label{prop:td-rec}
  The image of any recognizable term language by the derivation
  relation of a finite linear top-down term rewriting system is
  recognizable.
\end{proposition}

Top-down systems form a strict subfamily of \emph{generalized
  semi-monadic term rewriting systems} \cite{gv98}, which is itself a
strict subfamily of \emph{right-linear finite path overlapping
  systems} \cite{tks00}. Both classes have been proven to preserve
recognizability. As a consequence, this is also the case for top-down
systems. However, it should be mentioned that neither of these classes
has a rational derivation. Indeed, it is quite easy to find a
generalized semi-monadic system whose derivation cannot be recognized
by any finite mechanism. For instance, the generalized semi-monadic
system whose unique rule is $gx \ra{}{} fgfx$ clearly has a
non-rational derivation: its intersection with the rational relation
$ga \times f^*gf^*a$ is $ga \times \{f^ngf^na | n \geq 0\}$. By the
usual pumping arguments (adapted to this new setting), this relation
is not rational.

Finally, please note that the inverse of a top-down system is, by
definition, bottom-up. For any top-down system we can build a
grammar $G$ recognizing $\ra{R}{*}$. Thus, the grammar
$\pi_2(G)\pi_1(G)$ obtained by swapping both projections of $G$
generates the derivation $\ra{R^{-1}}{*}$ of the bottom-up system
$R^{-1}$. Inverse recognizability preservation follows.

\begin{corollary}
  \label{cor:asc}
  Every finite linear bottom-up term rewriting system $R$ has a
  rational derivation $\ra{R}{*}$, and the inverse image by
  $\ra{R}{*}$ of any recognizable term language is recognizable.
\end{corollary}

\section{Derivation of Suffix Systems}
\label{sec:suffix}

This section presents a study of the derivation relations of suffix
term rewriting systems. After introducing a property related to the
notion of \emph{suffix rewriting}, we show that the derivation of any
recognizable linear suffix system is rational. Finally, we prove that
the image or inverse image of any recognizable term language by the
derivation of a recognizable linear suffix system is recognizable, and
that it is possible to build a tree automaton accepting it.

\begin{definition}
  The suffix rewriting of a term rewriting system $R$ is the relation
  \begin{multline*}
    \esra{R}{}\ =\ \big\{\,(c[l\sigma],c[r\sigma]) \in T(F,X)^2\ |\ 
    (l,r) \in R\ \land\ c \in C_1(F)\\
    \land\ \sigma \in S(\emptyset,X) \mbox{ bijective} \big\}
  \end{multline*}
  (a bijective substitution in $S(\emptyset,X)$ is a bijective
  variable renaming over $X$).
\end{definition}

Suffix systems have a specific behaviour with respect to suffix
rewriting. Indeed, the derivation of any input tree $t$ by a suffix
system can always be decomposed in two phases. First, a prefix
$\bar{t}$ of $t$ is read, and several steps of suffix rewriting can be
applied to it. Once this first sequence is over, $\bar{t}$ has been
rewritten into a prefix $\bar{s}$ of $s$, never to be modified
anymore. In a second time, the rest of $t$ is derived in the same
fashion, starting with suffix rewriting of a prefix of the remaining
input. As a consequence, the derivation of a suffix system is
equivalent to its `iterated' suffix derivation.

\begin{lemma}
  \label{lem:suffix}
  For any suffix term rewriting system $R$,
  $$
  s\ \era{R}{+}\ t\ \iff\ \left\{
    \begin{gathered}
      \exists\ \bar{s},\bar{t} \in T(F,X),\ \sigma,\tau \in S(F,X)\ s
      = \bar{s}\sigma\ \land\ t = \bar{t}\tau
      \\
      \land\ \bar{s}\ \esra{R}{+}\ \bar{t}\ \land\ \forall x \in
      \mathit{Var}(\bar{s})\,\cap\,\mathit{Var}(\bar{t}),\ \sigma(x)\ 
      \era{R}{*}\ \tau(x).
    \end{gathered}
  \right.
  $$
\end{lemma}

Another interesting property is that, for any recognizable system, a
suffix rewriting sequence is always equivalent to a sequence in two
parts, where the first part only consumes suffix subterms of the input
term, and the second part only produces new suffix subterms in their
place.

\begin{lemma} 
  \label{lem:sat}
  For all recognizable linear term rewriting system $R$ over $F$ and
  $X$, there exist a finite ranked alphabet $Q$ and three finite
  rewriting systems
  \begin{itemize}
  \item $\rd\,\subseteq\,\{px \ra{}{} fp_1x_1 \ldots p_nx_n\ |\ f \in
    F,\ p,p_1, \ldots ,p_n \in Q,\ x,x_1, \ldots ,x_n \in X^*\}$
  \item $\rr\,\subseteq\,\{px \ra{}{} qy\ |\ p,q \in Q,\ x,y \in
    X^*\}$
  \item $\ru\,\subseteq\,\{fp_1x_1 \ldots p_nx_n \ra{}{} px\ |\ f \in
    F,\ p,p_1, \ldots ,p_n \in Q,\ x,x_1, \ldots ,x_n \in X^*\}$
  \end{itemize}
  such that $s\ \esra{R}{*}\ t\ \iff\ s\ \esra{\ru \cup \rr}{*}\ 
  \circ\ \esra{\rd \cup \rr}{*}\ t.$
\end{lemma}

Lemma \ref{lem:sat} can be reformulated in the following way: a pair
$(s,t)$ of terms belongs to the suffix derivation of a system $R$ if
and only if there is a context $c$ such that $s = c[s_1 \ldots s_n]$,
$t = c[t_1 \ldots t_n]$ and for all $i \in [1,n]$, there is a term
$q_ix_i$ such that $s_i \sra{\ru \cup \rr}{*} q_ix_i$ and $q_ix_i
\sra{\rd \cup \rr}{*} t_i$.

\begin{theorem}
  \label{th:suffix}
  Every recognizable linear suffix term rewriting system $R$ has a
  rational derivation.
\end{theorem}

\begin{proof}
  Let $R$ be a recognizable linear suffix system on $T(F,X)$. Let
  $\ru$, $\rr$ and $\rd$ be the rewriting systems mentioned in Lemma
  \ref{lem:sat}. Let $N$ be a set of pairs of the form $u|v$ where $u$
  and $v$ are two linear term words over $\mathit{Ran}(\ru \cup
  \rr)^*$ and $\mathit{Dom}(\rd \cup \rr)^*$ respectively. Note that
  $\mathit{Ran}$ and $\mathit{Dom}$ are defined up to a renaming of
  the variables. We can thus impose that $u$ and $v$ share the same
  set of variables ($\mathit{Var}(u) = \mathit{Var}(v)$), and there is
  no pair of strict subwords $u'$ and $v'$ of $u$ and $v$ such that
  $\mathit{Var}(u') \neq \mathit{Var}(v')$ (i.e one should not be able
  to split $u|v$ in two correct non-terminals). This, together with
  the facts that $F$ is finite and $u$ and $v$ are linear, implies
  that $N$ is finite for some fixed, standard variable renaming. Thus,
  given an axiom $I$, we can build a grammar $G$ whose set of
  non-terminals is $N\,\cup\,\{I,I'\}$, having the following finite
  sets of productions:
  \begin{multline}
    \label{eq:id}
    \forall\ f \in F,
    \\
    I\ \lra\ f I_1^1 \ldots I_1^n\ \times\ f I_2^1 \ldots I_2^n \quad
    \text{and} \quad I'\ \lra\ f I'{^1} \ldots I'{^n}
  \end{multline}
  \vspace{-5ex}
  \begin{multline}
    \label{eq:init}
    \forall\ px \in \mathit{Dom}(\rd \cup \rr) \cap \mathit{Ran}(\ru
    \cup \rr),
    \\
    I\ \lra\ px|px
  \end{multline}
  \vspace{-5ex}
  \begin{multline}
    \label{eq:eps}
    \forall\ u'\ \sra{\rr}{*}\ u,\quad v\ \sra{\rr}{*}\ v',\quad u'
    \in \mathit{Ran}(\ru)^*,\quad v' \in \mathit{Dom}(\rd)^*,
    \\
    u|v\ \lra\ u'|v'
  \end{multline}
  \vspace{-5ex}
  \begin{multline}
    \forall\ u_1 = p_1x_1 \ldots p_ix_i,\quad u_2 = p_{j+1}x_{j+1}
    \ldots p_nx_n,
    \\
    v = q_1y_1 \ldots q_my_m,\quad fp_{i+1}x_{i+1} \ldots
    p_jx_j\,\ru\,px,
    \\
    \label{eq:prodg}
    u_1\,px\,u_2|v\ \lra\ \mu_1 \ldots \mu_i\,(f\mu_{i+1} \ldots
    \mu_{j})\,\mu_{j+1} \ldots \mu_n\ \times\ \nu_1 \ldots \nu_m
  \end{multline}
  \vspace{-5ex}
  \begin{multline}
    \forall\ u = p_1x_1 \ldots p_nx_n,\quad v_1 = q_1y_1 \ldots
    q_iy_i,
    \\
    v_2 = q_{j+1}y_{j+1} \ldots q_my_m,\quad qy\,\rd\,fq_{i+1}y_{i+1}
    \ldots q_jy_j,
    \\
    \label{eq:prodd}
    u|v_1\,qy\,v_2\ \lra\ \mu_1 \ldots \mu_n\ \times\ \nu_1 \ldots
    \nu_i\,(f\nu_{i+1} \ldots \nu_{j})\,\nu_{j+1} \ldots \nu_m
  \end{multline}
  In rules (\ref{eq:prodg}) and (\ref{eq:prodd}), all the $(\mu_k)_{k
    \in [1,n]}$ and $(\nu_k)_{k \in [1,m]}$ are variables belonging to
  instances of non-terminals $u'|v' \in N$ where $u'$ and $v'$ are
  built from terms $(p_kx_k)_{k \in [1,n]}$ and $(q_ky_k)_{k \in
    [1,m]}$ respectively.  Variables $\mu_1$ to $\mu_n$ (resp. $\nu_1$
  to $\nu_m$) appear only in the first (resp. second) projection of
  any non-terminal. Note that this instantiation is unique, by
  construction of the set $N$. It is also always possible since every
  rule of $R$ is, by hypothesis, linear.
  
  Call $\rho$ the substitution which maps each non-terminal variable
  $(u|v)_i$ to the term $(u)_i$ if $i \in [1,|u|]$ and to $(v)_i$ if
  $i \in [|u|+1,|u|+|v|]$, and each non-terminal variable $(I^i_j)_{j
    \in [1,2]}$ to a variable $x_i$. It is clear from the rules of
  $G_0$ that:
  \begin{equation}
    \label{eq:prefpart2}
    I\ \era{G_0}{*}\ s\,\times\,t\ \iff\ s\rho\ \esra{\ru 
      \cup \rr}{*}\ \circ\ \esra{\rd 
      \cup \rr}{*}\ t\rho.
  \end{equation}
  We will not detail the proof of this observation. Notice that this
  grammar works in a very similar way to a \emph{ground tree
    transducer}, which is the formalism used by \cite{dt85} to
  recognize the derivation of a ground system. The only difference is
  that we keep track of the variables appearing in the left and right
  projections of the relation, so as to be able to resume the
  rewriting at relevant positions. Now add to $G_0$ the set of rules
  \begin{align}
    \forall\ x \in X \mbox{ such that } x R px,\ qx R x,
    \label{eq:reinit}
    & & px|qx\ &\lra\ I
    \\
    \label{eq:reinitbis}
    & & px|\ \quad &\lra\ I'
  \end{align}
  and call this new grammar $G$. These last rules allow the derivation
  to go on properly after a first sequence of suffix rewritings has
  taken place, by creating new instances of the axiom between leaves
  where the same variable would appear. By Lemma \ref{lem:suffix}, $G$
  generates $\ra{R}{*}$. \qed
\end{proof}

\begin{proposition}
  \label{prop:suffrec}
  The image and inverse image of any recognizable term language by the
  derivation of a finite linear suffix term rewriting system is
  recognizable.
\end{proposition}

\paragraph{Proof sketch.}
Once a grammar generating the derivation of a suffix system $R$ is
built, according to the previous proof, it is not difficult to
synchronize the left projection of this grammar with any finite
top-down tree automaton $A$. We thus obtain a new grammar, whose
second projection yields a finite automaton accepting the image of
$L(A)$ by $\ra{R}{*}$. This is symmetrical, hence the converse. \qed

We will illustrate the fact that suffix systems are strictly more
general than ground systems on the following simple example.

\begin{example}
  \label{ex:gtrs}
  Consider the finite suffix system $R = \{fxy\,\rightarrow\,fyx,\ a\,
  \rightarrow\,ga\}$ over the ranked alphabet $\{f^{(2)},\ g^{(1)},\ 
  a^{(0)}\}$. The first rule of $R$ allows to swap at any time both
  children of an $f$-node. This somehow expresses the commutativity of
  $f$. The derivation of $R$ (restricted for the sake of clarity to
  $(fg^*ag^*a)^2$) is the relation $\{(fg^mag^na, fg^nag^ma)\,|\,m,n
  \geq 0\} \cup \{(fg^mag^na, fg^{m+1}ag^na)\,|\,m,n \geq 0\} \cup
  \{(fg^mag^na, fg^mag^{n+1}a)\,|\,m,n \geq 0\}, $ which is not
  recognizable by a ground tree transducer.
  
  Furthermore, we claim that the transition graph of this rewriting
  system is not isomorphic to the transition system of any
  (recognizable) ground term rewriting system as defined in
  \cite{lo02,co02}. \emph{Note:} the transition graph of a rewriting
  system is the graph whose vertices are the terms from the domain or
  range of the system, and whose edges are all the pairs $(s,t)$ such
  that $s$ can be rewritten to $t$ in one step.
\end{example}

\section{Conclusion}

This paper extends the left, right, prefix and suffix word rewriting
systems defined in \cite{ca00} to bottom-up, top-down, suffix and
prefix term rewriting systems. The derivation relation of the three
first types of systems can be generated by finite graph grammars,
while systems of the fourth type have a non recursive derivation in
general. We also stated some recognizability preservation properties
of these classes of systems, and provided effective constructions in
each case. Although \cite{tks00} defines a class of
recognizability-preserving rewriting systems strictly more general
than top-down systems, they do not aim to provide a construction for
the derivation relation itself, which is indeed not rational. As for
suffix systems, to our knowledge, no comparable class of
recognizability-preserving systems has been defined yet.

This study puts in practical use the notion of rationality defined in
\cite{ra97}, which nicely extends the usual rational relations on
words, even though some of their key properties are missing, like the
closure by composition or systematic preservation of recognizability.
However, this formalism is an interesting and powerful work basis for
the study of binary relations on terms, especially thanks to the fact
that it is general enough to extend asynchronous transducers (which is
not the case of most other formalisms). Still, depending on one's
objectives, it might be necessary to devise a more restricted
notion of rational relations on terms, which would be closed under
composition or preserve recognizability (or both). Note that
\cite{ra97} contains the definition of such a subfamily of relations
(called rational \emph{transductions}). However, it can be shown that
the derivations of some top-down systems do not belong to this class.

Finally, it could be interesting to look for extensions to some of the
existing works previously mentioned. First, one may try to elaborate
actual verification methods using our systems to model transitions,
and recognizable term languages for sets of configurations, along the
ideas of regular model-checking \cite{bjnt00}. Indeed, being able to
effectively build the transitive closure of the system's transition
relation and compute the image of regular sets of configurations could
lead to interesting results.  Second, the definitions from \cite{lo02}
and \cite{co02} about transition graphs of ground systems, should
extend smoothly to the case of suffix systems. Thus, it would be
meaningful to determine whether part or all of their results extend to
this new family, and in particular whether the transition graphs of
suffix systems have a decidable first order theory with reachability.
Note that, as illustrated in Ex. \ref{ex:gtrs}, we suspect that the
transition graphs of suffix systems \emph{strictly} include the former
families of graphs.

\pagebreak

\appendix
\section{Appendix: Proof details}

We will present here the missing proof to Lemma \ref{lem:top-down}, a
construction which can be used as an alternative proof for Proposition
\ref{prop:td-rec}, the proofs to Lemmas \ref{lem:suffix} and
\ref{lem:sat} and a detailed construction for Proposition
\ref{prop:suffrec}.

\subsection{Derivation of Bottom-up and Top-down Systems}

\noindent
{\bf Lemma \ref{lem:top-down}.} {\itshape The relations of derivation
  and top-down derivation of any top-down term rewriting system $R$
  coincide:
  $$
  \era{R}{*}\ =\ \ehra{R}{*}.
  $$}

\begin{proof}
  By definition, $\hra{R}{*}\ \subseteq\ \ra{R}{*}$. It remains to
  prove the converse, by using the same technique as in \cite{ca00}.
  As $\ra{R}{*}\ =\ \ra{R - (X \times X)}{*}$, we may assume that $R\ 
  \cap\ (X \times X) = \emptyset$. Then, we can sort any derivation
  into a \emph{top-down} derivation, as defined above, by applying an
  alternative version of the bubble sort algorithm in which removal
  and addition of elements are allowed. To do so, we use the following
  inclusions:
  \begin{alignat}{5}
    \rw{u,v}{\geqslant p} & \ \circ\ \rw{x,v'}{p} & & \ \subseteq\ &
    \rw{x,v'}{p} & \ \circ\ \rw{u,v}{2^{>p}} & & & &
    \label{inc:1} \\
    \rw{u,v}{>p} & \ \circ\ \rw{u',v'}{p} & & \ \subseteq\ &
    \rw{u',v'}{p} & \ \circ\ \rw{u,v}{2^{\geqslant p}} & & & &
    \mbox{with}\ u,u' \notin X
    \label{inc:2} \\
    \rw{u,v}{>p} & \ \circ\ \rw{x,v'}{p} & & \ \subseteq\ &
    \rw{x,v'}{p} & \ \circ\ \rw{u,v}{2^{>p}} & & & & \mbox{with}\ 
    u',v' \notin X
    \label{inc:3} \\
    \rw{u,v}{>p} & \ \circ\ \rw{x,v'}{p} & & \ \subseteq\ &
    \rw{x,v'}{p} & \ \circ\ \rw{u,v}{2^{>p}} & & \quad \cup \quad & &
    \rw{u,v}{p} \ \circ\ \rw{x,v'}{2^{>p}}
    \label{inc:4} 
  \end{alignat}
  where $\ra{>p}{}$ (resp. $\ra{\geqslant p}{}$) denotes rewriting at
  any position greater than (resp. greater or equal to) $p$, and
  $\ra{2^{>p}}{}$ or $\ra{2^{\geqslant p}}{}$ denotes multi-step
  rewriting at any set of such positions. Let us now prove the first
  inclusion. Let
  $$
  r \rw{u,v}p s \rw{x,v'}q t
  $$
  with $p \geqslant q$. One can find 1-contexts $\bar{r}$ and
  $\bar{s}$ and substitutions $\rho$ and $\sigma$ such that $r\ =\ 
  \bar{r} [u \rho],\ s = \bar{r} [v \rho] = \bar{s} [x \sigma]$ and $t
  = \bar{s} [v' \sigma]$, with $\mathit{pos}(\ssq,\bar{r}) = {p}$ and
  $\mathit{pos}(\ssq,\bar{s}) = {q}$. As $p \geqslant q$ and by
  hypothesis on $R$, there is a substitution $\gamma$ such that
  $\bar{r} = \bar{s} [x \gamma]$ and $\sigma = \gamma[v\rho]$. It is
  thus possible to swap the rewriting steps between $r$ and $t$:
  $$
  r = (\bar{s} [x \gamma]) [u \rho] \rw{x,v'}q (\bar{s} [v'
  \gamma]) [u \rho] \rw{u,v}{P} (\bar{s} [v' \gamma]) [v \rho] = t,
  $$
  where $P\ \subseteq\ 2^{>q}$ is the set of positions at which the
  special variable $\ssq$ occurs in $\bar{s} [v' \gamma]$ ($P =
  \mathit{pos}(\ssq,\bar{s} [v' \gamma])$). Hence
  $$
  r \rw{x,v'}q \circ \rw{u,v}{2^{>q}} t,
  $$
  and inclusion (\ref{inc:1}) is proven. More generally, consider
  any two-steps rewriting
  $$
  r \rw{u,v}p s \rw{u',v'}q t
  $$
  with $p > q$. By definition there must exist 1-contexts $\bar{r}$
  and $\bar{s}$ and substitutions $\rho$ and $\sigma$ such that $r =
  \bar{r} [u \rho],\ s = \bar{r} [v \rho] = \bar{s} [u´ \sigma]$ and
  $t = \bar{s} [v' \sigma]$, with $\mathit{pos}(\ssq,\bar{r}) = \{p\}$
  and $\mathit{pos}(\ssq,\bar{s}) = \{q\}$. The hypotheses we made on
  $R$ imply certain restrictions on the structure of configuration
  $s$.  As $R$ is descending, $\bar{s} [u'] \leqslant \bar{r}$. So
  there must be a substitution $\gamma$ such that $\bar{r} = \bar{s}
  [u' \gamma]$ and $\sigma = \gamma[v\rho]$. Hence
  $$
  r = (\bar{s} [u' \gamma]) [u \rho] \rw{u',v'}q (\bar{s} [v'
  \gamma]) [u \rho] \rw{u,v}{P} (\bar{s} [v' \gamma]) [v \rho] = t,
  $$
  where $P = \mathit{pos}(\ssq,\bar{s} [v' \gamma])$. Thus
  $$
  r \rw{x,v'}q \circ \rw{u,v}{P} t.
  $$
  As $P\ \in\ 2^{\geqslant q}$, inclusion (\ref{inc:2}) is proven.
  If $v' \notin X$ or $v' = y \in X \land \gamma(y)$ is not trivial,
  then $P \in 2^{>q}$, so (\ref{inc:3}) and the first part of
  (\ref{inc:4}) are true.  Finally, if $u \in X,\ v' \in X$ and
  $\gamma(v')$ is trivial, then $r = \bar{r} \rho,\ t = \bar{s}
  \sigma,\ \bar{r} = \bar{s} [u'],\ v \rho = \sigma$, hence
  $$
  r = (\bar{s} [u']) \rho \rw{x,v}q (\bar{s} [v [u']]) \rho
  \rw{u',x'}{P \ \subseteq\ 2^{> q}} (\bar{s} [v]) \rho = t. \qed
  $$
\end{proof}

\noindent
{\bf Proposition \ref{prop:td-rec}.} {\itshape The image of any
  recognizable term language by the derivation of a linear top-down
  term rewriting system is recognizable.}

\begin{proof}
  First note that the domain of the derivation of a descending system
  is $T(F)$, and its range is recognizable: the grammar constructed in
  the proof of Theorem \ref{th:top-down} only has non-terminals whose
  second projection is reduced to a single variable.
  
  Let $L$ be a recognizable term language accepted by some top-down
  tree automaton $A$ with a set of control states $Q$. Let $R$ be a
  linear descending system. By the previous construction, we are able to
  build a rational grammar $G$ recognizing $\ra{R}{*}$. Let us now
  define a ``product grammar'' $G_A$ whose domain is $L$ and whose
  range is $\ra{R}{*}(L)$.  $G_A$'s non-terminals will be of the form
  \begin{equation}
    (N_1,q_1)\ldots(N_n,q_n)(N_{n+1},q_1 \ldots q_n), 
    \label{eq:nonterm}
  \end{equation}
  noted $N_{q_1 \ldots q_n}$ for short. For each production
  $$
  N_1 \ldots N_{n+1}\ \era{}{}\ t_1 \ldots t_n \ s
  $$
  of $G$, the product grammar $G_A$ will have all possible
  productions
  $$
  N_{q_1 \ldots q_n}\ \era{}{}\ t'_1 \ldots t'_n \ s'
  $$
  where the $m$-contextual term word $t_1 \ldots t_n$ is partially
  accepted by $A$ (see Sect. \ref{prelim}) with initial control word
  $(q_1 \ldots q_n)$ and final control word $(q'_1 \ldots q'_m)$.
  Furthermore, $t'_1 \ldots t'_n$ is obtained from $t_1 \ldots t_n$ by
  pairing each of its $m$ variables with the associated component of
  the final control word, and $s'$ is obtained from $s$ by pairing
  each of its variables with a word on $Q^*$, so as to complete every
  instance of $G_A$'s non-terminals according to (\ref{eq:nonterm}).
  The result of this is that a pair $(s,t)$ belongs to $L(G_A)$ if and
  only if $(s,t) \in L(G)$ and $s \in L$.  In other words, the second
  projection of $G_A$ is exactly the set of terms who are the image of
  some term in $L$ by $\ra{R}{*}$, so $\pi_2(L(G_A)) =\ \ra{R}{*}(L)$.
  By forgetting the left projection of every grammar production, one
  gets a grammar where all non-terminals have an arity of 1. Such
  grammars are called \emph{regular tree grammars} and generate
  recognizable term languages.  In this case, we obtain a regular
  grammar generating $\ra{R}{*}(L)$. \qed
\end{proof}

\subsection{Derivation of Suffix Systems}

\noindent
{\bf Lemma \ref{lem:suffix}} {\itshape For any suffix term rewriting
  system $R$,
  $$
  s\ \era{R}{k}\ t,\ k>0\ \iff\ \left\{
    \begin{gathered}
      \exists\ \bar{s},\bar{t} \in T(F,X),\ \sigma,\tau \in S(F,X),\ 
      k'>0,
      \\
      s = \bar{s}\sigma\ \land\ t = \bar{t}\tau \land\ \bar{s}\ 
      \esra{R}{k'}\ \bar{t}
      \\
      \land\ \forall x \in \mathit{Var}(\bar{s})\, \cap\,
      \mathit{Var}(\bar{t})\ \sigma(x)\ \era{R}{\leq k-k'}\ \tau(x).
    \end{gathered}
  \right.
  $$}

\begin{proof}
  First note that, since $\sra{R}{}\ \subseteq\ \ra{R}{}$, the inverse
  implication is trivial. It only remains to prove the direct
  implication. We will reason by induction on $k$:
  \begin{list}{}{}
  \item[$k = 1$: ] $s \ra{R}{} t$ implies that there is a context $c$
    and a substitution $\sigma$ such that $s = c[l\sigma]$ and $t =
    c[r\sigma]$.  Thus by definition of suffix rewriting $c[l]\ 
    \sra{R}\ c[r]$, and of course for all variable $x$, $\sigma(x)
    \ra{R}{0} \sigma(x)$.
  \item[$k \Rightarrow k+1$: ] let $s \ra{R}{k} s'\ \rw{l,r}{p} t$
    with $lRr$, $s = \bar{s}\sigma$ and $s' = \bar{s'}\sigma'$. Two
    cases:
    \begin{list}{}{}
    \item[$p \in \mathit{Pos}(\bar{s'})$: ] if there exists a context
      $c$ such that $\bar{s'} = c[l]$, then we have $\bar{s}\ 
      \sra{R}{k}\ \bar{s'}\ \sra{R}{}\ c[r]$ and the condition is
      verified with $k' = k$ and $t = (c[r])\sigma'$. If not, then
      there must exist a context $c$ and a non-trivial substitution
      $\omega'$ such that $\bar{s'} = (c[l])\omega'$. As $R$ is suffix
      and as, by induction hypothesis, $\bar{s} \sra{R}{*} \bar{s'}$,
      there must exist $\omega$ such that $\bar{s} = (c[l])\omega$ and
      for all variable $x$ common to $l$ and $r$, $\omega(x)
      \sra{R}{*} \omega'(x)$.  We can then write $s$ as
      $(c[l])\omega\sigma$, $t$ as $(c[r])\omega'\sigma'$, and verify
      the condition is true with $k' = 1$.
    \item[$p \notin \mathit{Pos}(\bar{s'})$: ] by induction
      hypothesis, one can find $k'>0$ such that for all $x$ common to
      $\bar{s}$ and $\bar{s'}$, $\sigma(x) \ra{R}{\leq k-k'}
      \sigma'(x)$.  Furthermore, by applying rule $(l,r)$ to one of
      the $\sigma'(x)$, we get $\sigma'(x) \ra{l,r}{} \tau(x)$. We
      thus have $t = \bar{s'}\tau$, $\bar{s} \sra{R}{k'} \bar{s'}$ and
      $\sigma(x) \ra{R}{k-k'+1} \tau(x)$ for all $x$ in both $\bar{s}$
      and $\bar{s'}$, which verifies the condition and concludes the
      proof. \qed
    \end{list}
  \end{list}
\end{proof}

\noindent
{\bf Lemma \ref{lem:sat}.} {\itshape For all recognizable linear term
  rewriting system $R$ over $F$ and $X$, there exist a finite ranked
  alphabet $Q$ and three finite rewriting systems
  \begin{itemize}
  \item $\ru\ \subseteq\ \{px \ra{}{} fp_1x_1 \ldots p_nx_n\ 
    |\ f \in F,\ p,p_1, \ldots ,p_n \in Q,\ x,x_1, \ldots ,x_n \in
    X^*\}$,
  \item $\rr\ \subseteq\ \{px \ra{}{} qy\ |\ p,q \in Q,\ 
    x,y \in X^*\}$,
  \item $\rd\ \subseteq\ \{fp_1x_1 \ldots p_nx_n \ra{}{} px\ |\ f \in
    F,\ p,p_1, \ldots ,p_n \in Q,\ x,x_1, \ldots ,x_n \in X^*\}$
  \end{itemize}
  such that
  $$
  s\ \esra{R}{*}\ t\ \iff\ s\ \esra{\ru \cup
    \rr}{*}\ \circ\ \esra{\rd \cup
    \rr}{*}\ t.
  $$}

\begin{proof}
  Let $R$ be a recognizable linear rewriting system over $F$ and $X$.
  Let $(A_i,B_i)_{i \in [1,l]}$ be $l$ pairs of finite top-down
  automata, each accepting a set of rules of $R$. The set of states of
  $A_i$ (resp. $B_i$) will be referred to as $P_i$ (resp.  $Q_i$).
  Without losing generality, we will suppose that all $P_i$ and $Q_i$
  are pairwise disjoint. We also define $P$ as the union of all $P_i$
  and $Q$ as the union of all $Q_i$. For all state $p \in P \cup Q$,
  define $\nu(p)$ as the set of all possible variable boundaries in
  the language accepted with initial state $p$.  The way we defined
  recognizable linear systems, i.e. with a finite number of variables,
  we can consider without losing generality that $|\nu(p)| = 1$ for
  all $p$.
  
  In a first step, we define a new rewriting system $R'$ on $F \cup P
  \cup Q$ and $X$. The state alphabets $P$ and $Q$ are considered as
  ranked alphabets, where the arity of any of their symbols is equal
  to the number of variables appearing in the terms of its associated
  term language (which can be supposed unique without losing
  generality). For instance, suppose that from some state $p$,
  automaton $A$ accepts the language $g^*fxy$. Then, $p$ will be
  considered as a binary symbol.

  We give $R'$ the following set of rules. For all rule $pf
  \rightarrow p_1 \ldots p_m$ in some $A_i$ such that $\nu(p) =
  \nu(p_1) \ldots \nu(p_m) = x_1 \ldots x_n$, we have:
  \[  
  f p_1 \nu(p_1) \ldots p_m \nu(p_m)\ \rightarrow\ p x_1 \ldots x_n\ 
  \in\ R'.
  \]
  Rules of this kind allow us to consume a left-hand side of a rule in
  the input tree. For all rule $qf \rightarrow q_1 \ldots q_m$ in
  some $B_i$ such that $\nu(q) = \nu(q_1) \ldots \nu(q_m) = x_1 \ldots
  x_n$, we have:
  \[  
  q x_1 \ldots x_n\ \rightarrow\ f q_1 \nu(q_1) \ldots q_m \nu(q_m)\ 
  \in\ R'.
  \]
  Rules of this kind allow us to produce a right-hand side of a rule
  whose left-hand side has been previously consumed. Finally, for all
  pair $(p_0,q_0)$ of initial states of some pair $(A_i,B_i)$, with
  $\nu(p_0) = x_1 \ldots x_n$ and $\nu(q_0) = x_{k_1} \ldots x_{k_m}$
  we have:
  \[
  p_0 x_1 \ldots x_n\ \rightarrow\ q_0 x_{k_1} \ldots x_{k_m}\ \in\ 
  R'.
  \]
  This simulates the application of a rewrite rule from $L(A_i) \times
  L(B_i)$ by initiating a run of automaton $B_i$ when a successful
  `reverse run' of $A_i$ has been achieved. When restricted to
  $T(F)^2$, the derivation of $R'$ coincides with $\ra{R}{*}$:
  \begin{equation}
    \label{eq:trs_equiv2}
    \forall s,t \in T(F),\ s\ \era{R}{*}\ t\ \iff\ s\ \era{R'}{*}\ t.
  \end{equation}
  The proof of this property is not difficult and will thus not be
  detailed here. In the rest of the proof, $p,q$ and all variations
  thereof designate automata control states in $P \cup Q$, variable
  words in $X^*$ are denoted by $u,v,u_i,v_i, \ldots$, and $\sigma$ is
  a variable renaming.
  
  In a second step, we define $\rd$ as $\{fp_1u_1 \ldots
  p_nu_n\,R'\,pu\}$ (`consuming' rules), $\ru$ as $\{qv\,R'\,fq_1v_1
  \ldots q_nv_n\}$ (`producing' rules) and $\rr$ as the smallest
  binary relation in $T(F,X)^2$ closed by the following inference
  rules:
  \begin{gather*}
    \frac{}{pu\ \rr\ pu} \quad(1) \qquad \frac{pu\ R'\ 
      qv}{pu\ \rr\ qv} \quad(2)
    \\ \\
    \frac{pu\ \rr\ qv}{pu\sigma\ \rr\ 
      qv\sigma} \quad (3) \qquad \frac{qu\ \rr\ q'v \quad
      q'v\ \rr\ q''z}{qu\ \rr\ q''z} \quad(4)
    \\ \\
    \frac{pu\ \rd\ fp_1u_1 \ldots p_nu_n \quad fq_1v_1
      \ldots q_nv_n\ \ru\ qv \quad \forall i,\,p_iu_i\ 
      \rr\ q_iv_i}{pu\ \rr\ qv} \quad (5)
  \end{gather*}
  Since $F,P,Q$ and $X$ are finite, and each symbol $p$ of $P \cup Q$
  has a definite arity, then $\rr$ is finite and effectively
  computable. Let us mention a simple property of $\rr$:
  \begin{equation}
    \label{eq:pont}
    \forall\,pu,t,qv \in T(F,X),\ pu\ \esra{\rd \cup
      \rr}{*}\ t\ \esra{\ru \cup 
      \rr}{*}\ qv\ \Lra\ pu\,\rr\,qv.
  \end{equation}
  This can be proved by induction on the nesting depth $k$ of term
  $t$:
  \begin{list}{}{}
  \item[$k = 0$: ] no rule of $\rd$ or $\ru$ is
    applied, thus $pu \sra{\rr}{*} qv$. Then, by inference
    rule (4), the property is true.
  \item[$k \Rightarrow k+1$: ] let us decompose the derivation
    sequence between $s$ and $t$:
    \begin{gather*}
      pu\ \esra{\rr}{}\ p'u'\ \esra{\rd}{}\ fp_1u_1 \ldots p_nu_n
      \\
      \esra{\rd \cup \rr}{*} \circ \esra{\ru \cup \rr}{*}\ fq_1v_1
      \ldots q_nv_n\ \esra{\ru}{}\ q'v'\ \esra{\rr}{} qv.
    \end{gather*}
    By induction hypothesis, we know that $fp_1u_1 \ldots p_nu_n
    \sra{\rr}{*} fq_1v_1 \ldots q_nv_n$, so by inference
    rule (5), $p'u'\,\rr\,q'v'$.  Finally, by inference
    rule (4), $pu\,\rr\,qv$.
  \end{list}
  It remains to prove that $\ru$, $\rr$ and $\rd$ verify the lemma. By
  (\ref{eq:trs_equiv2}), it suffices to show that $\sra{\ru \cup
    \rr}{*}\ \circ\ \sra{\rd \cup \rr}{*}\ =\ \sra{R'}{*}$. Proving
  the direct inclusion is equivalent to proving $\rr\ \subseteq\ 
  \sra{R'}{*}$, which can be seen easily by observing the inference
  rules given above. To prove the converse we will establish, by
  induction on $k$, the following property:
  \begin{gather*}
    \forall k,\ s\ \esra{R'}{k}\ t\ \Lra\ \exists\ c \in C_1(F),\ 
    (q_iu_i)_{i \in [n]} \in T(F,X),
    \\
    s\ \esra{\ru \cup \rr}{*}\ c[q_1u_1 \ldots q_nu_n]\ \esra{\rd \cup
      \rr}{*}\ t.
  \end{gather*}
  \begin{list}{}{}
  \item[$k = 0$: ] $s = t = c$, $n = 0$, so trivially $s\ 
    \sra{\ru \cup \rr}{*}\ c\ 
    \sra{\rd \cup \rr}{*}\ t$.
  \item[$k \Rightarrow k+1$: ] let $s\ \sra{R'}{k}\ t\ \sra{l,r}{} t'$
    with $l R' r$. By induction hypothesis,
    $$
    s\ \sra{\ru \cup \rr}{*}\ c[q_1u_1 \ldots
    q_nu_n]\ \sra{\rd \cup \rr}{*}\ t = c[t_1
    \ldots t_n].
    $$
    There are three possible cases:
    \begin{itemize}
    \item $\exists\,c',\ t = c'[t_1 \ldots t_n\,l]$: in this case $s$
      can be written $c'[s_1 \ldots s_n\,l]$, so the following
      derivation sequence is valid:
      \begin{align*}
        s\ \esra{\ru \cup \rr}{*} &\ c'[q_1u_1
        \ldots q_nu_n\,l]\,\ \esra{\ru}{}\,\ c'[q_1u_1
        \ldots q_nu_n\,q_lu]
        \\
        \esra{\rr}{}\,\ &\ c'[q_1u_1 \ldots q_nu_n\,q_rv]\ 
        \mbox{(because $l\,R'\,r$)}
        \\
        \esra{\rd}{}\,\ &\ c'[q_1u_1 \ldots q_nu_n\,r]\ 
        \esra{\rd \cup \rr}{*}\ c'[t_1 \ldots
        t_n\,r] = t'.
      \end{align*}
    \item $\exists\,i,\ t_i = \bar{t_i}[l]$: the only way to produce
      $t_i$ from $q_iu_i$ is along the steps
      $$
      q_iu_i \sra{\rd \cup \rr}{*}
      \bar{t_i}[q_lu] \sra{\rd}{}\ t_i.
      $$
      Thus the following derivation is valid:
      \begin{align*}
        s\ \esra{\ru \cup \rr}{*}\ c[q_1u_1
        \ldots q_nu_n]\ \esra{\rd \cup \rr}{*}
        &\ c[q_1u_1 \ldots \bar{t_i}[q_lu] \ldots q_nu_n]
        \\
        \esra{\rr}{}\,\ &\ c[q_1u_1 \ldots \bar{t_i}[q_rv]
        \ldots q_nu_n]\ \esra{\rd}{}\ t'.
      \end{align*}
    \item $\exists\,c',\,j > i \geq 1,\ t = c'[t_1 \ldots
      t_i\,l\,t_{j+1} \ldots t_n]$, with $l = \bar{l}[t_{i+1} \ldots
      t_j]$. The derivation sequence between $s$ and $t'$ can then be
      written
      \begin{align*}
        s\ \esra{\ru \cup \rr}{*} &\ c[q_1u_1
        \ldots q_nu_n]\ \esra{\rd \cup \rr}{*}\ 
        c[t_1 \ldots t_n]
        \\
        \esra{\ru}{*}\,\ &\ c[t_1 \ldots
        t_i\,q'_{i+i}u'_{i+1} \ldots q'_ju'_j\,t_{j+1} \ldots t_n]
        \\
        \esra{\ru}{}\,\ &\ c'[t_1 \ldots t_i\,q_lu\,t_{j+1}
        \ldots t_n]\,\ \ \esra{\rr}{}\,\ \ c'[t_1 \ldots
        t_i\,q_rv\,t_{j+1} \ldots t_n]
        \\
        \esra{\rd}{}\,\ &\ c'[t_1 \ldots t_i\,r\,t_{j+1}
        \ldots t_n] = t'
      \end{align*}
      Notice that for all $k \in [i+1,j]$, $q_ku_k \sra{\rd
        \cup \rr}{*} t_k \sra{\ru \cup
        \rr}{*} q'_ku'_k$.  Thus, by (\ref{eq:pont}),
      $q_ku_k\,\rr\,q'_ku'_k$, hence
      \begin{align*}
        s\ \esra{\ru \cup \rr}{*}\ c[q_1u_1
        \ldots q_nu_n]\ \esra{\ru \cup \rr}{*}
        &\ c'[q_1u_1 \ldots t_i\,q_lu\,q_{j+1}u_{j+1} \ldots q_nu_n]
        \\
        \esra{\rd \cup \rr}{*} &\ c'[t_1 \ldots
        t_i\,r\,t_{j+1} \ldots t_n] = t'.
      \end{align*}
    \end{itemize}
  \end{list}
  This concludes the proof of lemma \ref{lem:sat}. \qed
\end{proof}

\noindent
{\bf Proposition \ref{prop:suffrec}.} {\itshape The image and inverse
  image of any recognizable term language by the derivation of a
  recognizable linear suffix term rewriting system is recognizable.}

\begin{proof}
  Let $R_0$ be a recognizable linear suffix term rewriting system on
  $T(F)$, $R$ an equivalent normalized system on $T(F \cup F')$, with
  $F'$ a set of new function symbols. Let $G$ be the grammar
  recognizing its derivation, built as in the proof to Theorem
  \ref{th:suffix}, $N \cup \{I,I'\}$ its set of non-terminals. Let $A$
  be a finite non-deterministic top-down tree automaton accepting a
  recognizable language $L$.  Suppose $Q_A$ is the set of control
  states of $A$, disjoint from $F$ and $X$, $q_0 \in Q_A$ its unique
  initial state. Let $Q_A'$ be a disjoint copy of $Q_A$, we define the
  following grammar $G_A$ having non-terminals in $Q_A \cup Q_A' \cup
  (Q_A \times \mathit{Ran}(\ru \cup R'))^* | \mathit{Dom}(\rd \cup
  R')^*$ and the following set of production rules:
  \begin{itemize}
  \item For all rule $rf \ra{A}{} fr_1 \ldots r_n$:
    \begin{align}
      r\ \lra\ & f (r_1)_1 \ldots (r_n)_1\ \times\ f (r_1)_2 \ldots
      (r_n)_2
      \\
      r'\ \lra\ & f r'_1 \ldots r'_n
    \end{align}
  \item For all $r \in Q_A$, $px \in \mathit{Dom}(R) \cap
    \mathit{Ran}(R)$:
    \begin{equation}
      r\ \lra\ (r,px)|px,
    \end{equation}
  \item For all rule $p_1x_1 \ldots p_nx_n|v\ \lra\ p'_1x'_1 \ldots
    p'_nx'_n|v'$ of type (\ref{eq:eps}) in $G$ and state word $r_1
    \ldots r_n \in Q_A^*$:
    \begin{equation}
      (r_1,p_1x_1) \ldots (r_n,p_nx_n)|v\ \lra\ (r_1,p'_1x'_1) 
      \ldots (r_n,p'_nx'_n)|v'
    \end{equation}
  \item For all rule $u_1\,px\,u_2|v$ $\ra{}{}$ $s_1 (f\mu_{i+1}
    \ldots \mu_j) s_2 \times t$ of type (\ref{eq:prodg}) of $G$ and
    states word $r_1 \ldots r_i\,r\,r_{j+1} \ldots r_n \in Q_A^*$:
    \begin{equation}
      (r_1,p_1) \ldots (r_n,p_n)|v\ \era{}{}\ s' \times t'
    \end{equation}
    where $rf \ra{A}{} fr_{i+1} \ldots r_{j}$ and $s'$ and $t'$ are
    obtained from $s$ and $t$ by replacing each occurrence of $p_k$ in
    a non-terminal variable by $(r_k,p_k)$.
  \item For all rule $p_1x_1 \ldots p_nx_n|v \ra{}{} s \times t$ of
    type (\ref{eq:prodd}) of $G$ and word $r_1 \ldots r_n \in Q_A^n$:
    \begin{equation}
      (r_1,p_1) \ldots (r_n,p_n)|v\ \era{}{}\ s' \times t'
    \end{equation}
    where $s'$ and $t'$ are obtained from $s$ and $t$ by replacing
    each occurrence of $p_i$ in a non-terminal variable by
    $(r_i,p_i)$.
  \item Finally, for all $x R px$, $qx R x$, $x \in X$ and $r \in
    Q_A$:
    \begin{align}
      (r,px)|qx\ \lra\ & r\\
      (r,px)|\ \quad \lra\ & r'.
    \end{align}
  \end{itemize}
  One can see that, starting from non-terminal $q_0$, grammar $G_A$
  only accepts pairs $(s,t)$ such that $s \ra{R}{*} t$ and $s \in L$.
  Thus the set of all $t$ such that $(s,t)$ is generated by $G_A$ from
  $q_0$ is exactly the image of $L$ by the derivation of $R$:
  $$
  L(G_A,q_0)\ =\ \era{R}{*}(L).
  $$
  An automaton recognizing this set of terms can be built by taking
  the right projection of $G_A$, and by treating each non-terminal
  variable as a unary non-terminal. The rules of this automaton are
  given by the rules of $G_A$ broken into several rules over these
  non-terminals of length 1.
  
  Note that this proof is totally symmetrical, and that the
  synchronization of $G$ by a finite automaton $A$ could be done on
  the second projection instead of the first. Thus the converse
  result. \qed
\end{proof}

\end{document}